\documentclass[twocolumn,superscriptaddress,aps]{revtex4-2}
\usepackage{graphicx}
\usepackage{amsmath}
\usepackage{amsthm}
\usepackage{amssymb}
\usepackage{latexsym}
\usepackage{array}
\usepackage{booktabs}
\usepackage{tabularx}
\usepackage{hyperref}
\usepackage{amsfonts}
\usepackage{dsfont}
\usepackage{mathrsfs}
\usepackage{verbatim}
\usepackage{bbold}
\usepackage{lipsum}
\usepackage[normalem]{ulem}
\usepackage{algorithm}

\usepackage{bm}
\usepackage{times}

\usepackage{upgreek}

\usepackage{makecell}
\usepackage{adjustbox}

\usepackage{algpseudocode}

\usepackage{color}

\newcommand{\bra}[1]{\left<#1\right|}
\newcommand{\ket}[1]{\left|#1\right>}
\newcommand{\abs}[1]{\left|#1\right|}
\newcommand{\norm}[1]{\left\lVert#1\right\rVert}

\newcommand{\braket}[2]{\left<{#1}|{#2}\right>}
\newcommand{\ketbra}[2]{\ket{#1}\!\!\bra{#2}}

\newcolumntype{L}[1]{>{\raggedright\arraybackslash}p{#1}}

\newtheorem{theorem}{Theorem}

\newtheorem{corollary}{Corollary}
\newtheorem{definition}{Definition}

\newcommand{\Tr}{\operatorname{Tr}}

\newcommand{\proj}{\mathbb{P}}
\newcommand{\eps}{\varepsilon}

\newcommand{\cH}{\mathcal{H}}

\newcommand{\cB}{\mathcal{B}}

\newcommand{\Id}{\hat{\mathds{1}}}

\newcommand{\hP}{\hat{P}}
\newcommand{\hrho}{\hat{\rho}}
\newcommand{\hM}{\hat{M}}
\newcommand{\hU}{\hat{U}}
\newcommand{\hV}{\hat{V}}
\newcommand{\hW}{\hat{W}}
\newcommand{\hQ}{\hat{Q}}
\newcommand{\hO}{\hat{O}}
\newcommand{\hPi}{\hat{\Pi}}

\newcommand{\Haar}{\operatorname{Haar}}
\newcommand{\Risk}{\operatorname{Risk}}

\begin{document}

\title{No Reference-Free Generalization in Quantum Machine Learning}

\author{Jeongho~Bang}\email{jbang@yonsei.ac.kr}
\affiliation{Institute for Convergence Research and Education in Advanced Technology, Yonsei University, Seoul 03722, Republic of Korea}
\affiliation{Department of Quantum Information, Yonsei University, Incheon 21983, Republic of Korea}

\begin{abstract}
Quantum machine learning is often motivated by the exponentially large state space of quantum systems, but this promise leaves a basic generalization problem unresolved: how can a learner assign different meanings to unseen quantum directions when the training data provide no preferred basis, measurement frame, or other orienting structure? We address this identifiability problem by formulating supervised learning without an external quantum reference frame, so that predictions cannot depend on an arbitrary choice of Hilbert-space coordinates. This requirement forces the learned classifier to preserve every unitary symmetry left unbroken by the training data. We prove that whenever the training states fail to span the full Hilbert space, all pure states orthogonal to their span must receive the same prediction---even when those states are mutually orthogonal and perfectly distinguishable once an appropriate measurement is supplied. The limitation is therefore not caused by state discrimination, optimization, or computational power, but by missing reference information. We further establish a robust version under weak symmetry breaking and show that learning generic unstructured concepts on multiqubit systems requires exponentially many independently oriented training directions. Numerical illustrations visualize the resulting prediction collapse and its controlled relaxation. Our results identify feature maps, measurement bases, Hamiltonians, locality, symmetry priors, architectures, and sufficiently diverse training states as operational resources for generalization. The central implication is that Hilbert-space dimension alone is not a learnable feature space: successful QML must specify the physical structure that gives unseen quantum directions semantic meaning.
\end{abstract}

\maketitle

\section{Introduction}\label{sec:introduction}

Quantum machine learning (QML) is often introduced through a geometric promise: an $n$-qubit system lives in a Hilbert space of dimension $d=2^n$, suggesting that a quantum model may have access to an exponentially large feature space~\cite{biamonte2017quantum,schuld2019feature,havlicek2019supervised}. The size of the ambient space, however, does not by itself explain how a learner should generalize. Generalization requires more than many available directions; it requires a rule that tells the learner which directions, distances, or transformations are relevant to the task. In practical QML, this semantic structure is usually supplied by a computational basis, a classical-to-quantum feature map, a known observable, a Hamiltonian, locality, a tensor-product decomposition, a symmetry prior, or a chosen variational architecture~\cite{schuld2020circuit,schuld2021encoding,cerezo2021variational,larocca2022group,meyer2023exploiting}. The motivating question of this work is what remains learnable when none of these structures is available and the learner must rely on the labeled quantum states alone.

A simple thought experiment reveals the difficulty. Suppose a learner is given a finite collection of labeled pure states, but no calibrated basis, feature map, measurement convention, or other side information. An external observer may name two previously unseen orthogonal states and assign them different labels. Quantum mechanics poses no obstacle to distinguishing those states once the appropriate measurement basis is supplied. Yet, the training data may contain no physical information that selects that basis or associates either direction with a particular label. A simultaneous unitary rotation of all quantum states then changes only the Hilbert-space coordinates used to describe the experiment, not the learning problem itself. A learner with no external reference should therefore give physically equivalent answers in every such description.

The purpose of this work is to turn that observation into a precise limitation on supervised quantum generalization. We isolate the orientational information carried by a quantum training set and ask how far this information can support predictions beyond the directions represented in the data. The resulting model is intentionally not a description of every QML algorithm. Most practical QML models are referenceful, because their encodings, measurements, hardware, or architectures already privilege particular directions. Our reference-free setting instead provides a controlled baseline: it asks what the labeled quantum examples can accomplish without those additional resources, and thereby identifies what practical models must supply in order to escape the limitation.

We call a learner reference-free when its prediction rule does not depend on an arbitrary choice of Hilbert-space coordinates: rotating all training states must simply rotate the learned classifier in the same way. This is the natural consistency requirement when no external quantum reference frame is available~\cite{bartlett2007reference,gour2008resource}. It also means that the learner cannot break a symmetry that the training data themselves leave intact. By placing no restriction on the optimizer, circuit ansatz, or computational power, we isolate a more basic question---whether the data make the target labels operationally identifiable at all.

Our main result is that, when the training states span only a proper subspace, every pure state in the orthogonal complement must receive the same prediction. The obstacle is therefore not quantum state discrimination, but the absence of a label-bearing reference for unseen directions. This conclusion remains stable under weak symmetry breaking. For generic binary concepts, it also yields an exponential barrier: on $n$ qubits, nontrivial worst-case generalization requires a training span whose rank is of order $2^n$. Repeated copies of already observed directions may improve estimation, but they do not supply the missing orientation.

These results complement several established limitations in QML. Barren plateaus concern the optimization~\cite{mcclean2018barren,cerezo2021variational}, quantum no-free-lunch theorems concern the average performance over broad task ensembles~\cite{poland2020nofree,sharma2022reformulation}, and capacity analyses ask how a specified model generalizes~\cite{caro2022generalization,abbas2021power,haug2024generalization,peters2023generalization,gilfuster2024rethinking}. Our theorem asks an earlier identifiability question: has the learning problem supplied enough physical structure for unseen directions to carry different labels at all? It also differs from the unknown-state discrimination, because the obstruction can persist for orthogonal states that are easy to distinguish once a measurement basis is supplied.

Taken together, the results support a resource-based view of quantum generalization. A training set carries not only statistical information in its labels, but also orientational information in the quantum directions it fixes. A small or low-rank dataset is therefore a partial quantum reference frame. Any remaining unitary freedom forces predictions to be constant along entire orbits unless a feature map, basis, observable, Hamiltonian, locality structure, symmetry representation, or architecture breaks that freedom. These familiar ingredients are not incidental implementation details; they are the physical and inductive structures that turn Hilbert-space directions into learnable features.

The resulting message can be summarized as follows:
\begin{quote}
\centering
``\emph{quantum generalization is symmetry breaking}.''
\end{quote}
The exponential dimension of Hilbert space is a reservoir of possible representations, not an automatic source of exponentially many learnable labels. This paper makes that distinction precise through an exact stabilizer-orbit no-go theorem, a robust approximate extension, and a quantitative reference-rank lower bound. The conclusion is that any successful generalization claim should identify the reference structure that gives operational meaning to unseen quantum directions.

\section{Learning without an external quantum reference frame}\label{sec:setup}

\subsection{Quantum supervised data and classifiers}

Let $\cH$ be a complex Hilbert space with finite dimension $d\ge 2$. Pure states are represented by rank-one projectors
\begin{eqnarray}
\hP_\psi:=\ketbra{\psi}{\psi},
\quad
\braket{\psi}{\psi}=1,
\label{eq:pure-projector}
\end{eqnarray}
so that the global phase of $\ket{\psi}$ is not physical~\cite{bang2026quantum}. A binary quantum training set is a finite list
\begin{eqnarray}
D=\{(\hP_i,y_i)\}_{i=1}^m,
\label{eq:dataset}
\end{eqnarray}
where $y_i\in\{0,1\}$. We allow repeated examples, but only the rank of their span will matter for the strongest statements. Throughout this work, the projectors are treated as the ideal mathematical description of the training states. This deliberately removes finite-copy tomography from the discussion. If only finitely many copies are available, one obtains additional estimation errors; those errors are not the source of the no-go theorem below. The corresponding finite-copy remarks, together with the extension to randomized learners, are given in Appendix~\ref{app:randomized}.

A trained binary classifier is represented by a POVM effect
\begin{eqnarray}
0\le \hM_D\le \Id.
\label{eq:effect}
\end{eqnarray}
The probability of predicting label $1$ on an input state $\hrho$ is
\begin{eqnarray}
p_D(1 | \hrho) = \Tr(\hM_D \hrho),
\label{eq:probability}
\end{eqnarray}
and the probability of label $0$ is $1-p_D(1|\hrho)$. This includes projective classifiers, general POVM classifiers, quantum kernel decision rules after training, and variational classifiers once the trained circuit and final measurement have been compressed into an effective two-outcome measurement~\cite{banchi2021generalization,schuld2020circuit,havlicek2019supervised}. For example, a variational classifier of the form
\begin{eqnarray}
\hrho \longmapsto \Tr\left[\hO \hU(\theta_D)\hrho\hU(\theta_D)^\dagger \right]
\label{eq:effective-measurement-example}
\end{eqnarray}
is represented by the effect $\hM_D=\hU(\theta_D)^\dagger \hO \hU(\theta_D)$, up to the standard affine rescaling needed to make it a POVM element. Thus, the theorems to be derived later do not depend on the training algorithm used to obtain $\hM_D$.

The model is intentionally generous. We do not impose computational limits, optimization limits, architectural restrictions, or sample-complexity limits. We simply ask what any final classifier can do if the learning rule is reference-free. Therefore, the no-go result is not a statement about the weakness of a particular optimizer or ansatz; it is a structural statement about the physical information contained in the quantum data.

\subsection{Reference-free covariance}

A global change of quantum reference frame is represented by a unitary $\hU\in U(\cH)$. It transforms a dataset as
\begin{eqnarray}
\hU D\hU^\dagger := \{(\hU\hP_i\hU^\dagger,y_i)\}_{i=1}^m.
\label{eq:transformed-data}
\end{eqnarray}
The labels are not changed, because they are classical outcomes or class names. Only the quantum representatives of the examples are rotated.

\begin{definition}[Reference-free learner]\label{def:reference-free}
A learning rule $D\mapsto \hM_D$ is reference-free if, for every unitary $\hU\in U(\cH)$,
\begin{eqnarray}
\hM_{\hU D\hU^\dagger}=\hU\hM_D\hU^\dagger.
\label{eq:covariance}
\end{eqnarray}
\end{definition}

Definition~\ref{def:reference-free} says that the learning rule has no built-in direction in Hilbert space. If all quantum states in the problem are rotated by $\hU$, the learned measurement rotates by the same $\hU$. This is the same physical logic used in quantum reference-frame theory: when no external frame is available, admissible descriptions and operations must respect the relevant symmetry~\cite{bartlett2007reference,gour2008resource}. In the present setting, the relevant symmetry is the full unitary group acting on the quantum feature Hilbert space.

The definition should not be confused with a claim that practical QML algorithms must be reference-free. Most practical QML algorithms are explicitly referenceful. A computational basis, a Pauli measurement, a fixed hardware coupling graph, a feature map $x\mapsto \ket{\phi(x)}$, or a Hamiltonian used to define the task already supplies a reference. Such algorithms can break Eq.~\eqref{eq:covariance}, and our theorem then identifies precisely what resource they are using: an external orientation of Hilbert space.

The covariance condition is also not merely a mathematical convenience. Suppose two laboratories describe the same physical training experiment using two different Hilbert-space coordinate systems related by $\hU$. If the learner has no additional reference, these two descriptions should lead to physically equivalent classifiers. Eq.~\eqref{eq:covariance} is exactly that consistency requirement. A rule that violates it must contain hidden coordinate information, and that coordinate information is the physical resource.

\subsection{Reference structure as a physical resource}\label{subsec:reference-resource}

The word ``reference'' is used here in a broad but operational sense. It means any structure, available to the learner before prediction, that selects preferred directions or preferred transformations in Hilbert space. A reference can be a physical device, such as, a measurement apparatus aligned to a Pauli axis. It can be an encoded convention, such as, a circuit that prepares a feature state from a classical input. It can also be a model prior, such as, a locality assumption or symmetry representation. What matters is that the structure is not inferred from the unlabeled complement of the training span; it is supplied independently of the unknown test direction~\cite{cong2019qcnn,nguyen2024equivariant}.

Table~\ref{tab:reference-resources} lists common resources and explains how each one breaks the global unitary freedom. The table is not meant to be exhaustive. Its purpose is to make clear that the theorem does not prohibit ordinary QML practice. Instead, it asks that the reference structure enabling generalization be stated explicitly. In many successful models, this reference structure is the most important part of the learning problem.

\begin{table*}[t]
\begin{adjustbox}{max width=0.98\textwidth}
\begin{tabular}{lll}
\toprule
Resource & What it fixes & How it avoids the no-go theorem \\
\midrule
\parbox{0.23\textwidth}{\raggedright Computational basis} & \parbox{0.36\textwidth}{\raggedright A preferred orthonormal basis $\{\ket{x}\}$} & \parbox{0.52\textwidth}{\raggedright Labels may refer to basis strings rather than arbitrary rays.} \\
\parbox{0.23\textwidth}{\raggedright Feature map $x\mapsto\ket{\phi(x)}$} & \parbox{0.36\textwidth}{\raggedright A structured image of classical data inside $\cH$} & \parbox{0.52\textwidth}{\raggedright Generalization is over the image of the map, not all of $\proj(\cH)$.} \\
\parbox{0.23\textwidth}{\raggedright Known observable or POVM} & \parbox{0.36\textwidth}{\raggedright Measurement directions and operational outcomes} & \parbox{0.52\textwidth}{\raggedright The classifier may depend on the supplied measurement frame.} \\
\parbox{0.23\textwidth}{\raggedright Hamiltonian or locality prior} & \parbox{0.36\textwidth}{\raggedright A tensor-product and interaction structure} & \parbox{0.52\textwidth}{\raggedright Only concepts compatible with the physical structure are considered.} \\
\parbox{0.23\textwidth}{\raggedright Symmetry prior} & \parbox{0.36\textwidth}{\raggedright A representation of a relevant task symmetry} & \parbox{0.52\textwidth}{\raggedright The model uses the intended symmetry, not the residual gauge symmetry.} \\
\parbox{0.23\textwidth}{\raggedright High-rank training set} & \parbox{0.36\textwidth}{\raggedright Many independent quantum directions} & \parbox{0.52\textwidth}{\raggedright The data themselves act as a partial quantum reference frame.} \\
\bottomrule
\end{tabular}
\end{adjustbox}
\caption{Typical sources of Hilbert-space reference structure in QML. A reference-free learner has none of these resources except the labeled projectors in $D$.}
\label{tab:reference-resources}
\end{table*}

This resource interpretation also explains why the reference-free assumption is physically meaningful even though it is restrictive. A theorem about reference-free learners is a theorem about what the data alone can do. Once a basis, observable, or feature map is added, one has changed the resource theory of the problem. The learner is no longer trying to infer labels from the projectors alone; it is using an additional physical convention.

\subsection{Scope and assumptions of the model}\label{subsec:scope}

The theorem uses a deliberately idealized model. This idealization is useful because it removes several familiar difficulties and leaves only the reference-frame obstruction. The assumptions are the following:
\begin{itemize}
\item The training states are treated as exact projectors. This means that we are not proving a tomography lower bound. If the projectors are unknown and only finitely many copies are supplied, then the learner also faces ordinary estimation noise. Such noise can only make the problem harder. The theorem asks what remains impossible even after this statistical difficulty is removed.

\item The output of training is allowed to be any POVM effect. Thus, we do not assume a particular circuit ansatz, optimizer, loss function, or hardware connectivity. A realistic QML algorithm may be much more constrained, but those constraints are not needed for the result. The no-go theorem applies to the most powerful possible post-training classifier compatible with the reference-free condition.

\item The reference-free condition is imposed on the learning rule, not on the target concept as imagined by an external observer. An observer may write down a target function using a basis. The theorem says that if this basis is not supplied to the learner by either data or side information, then a reference-free learning rule cannot reconstruct it. This distinction is important: the target may be mathematically well-defined in the author's notation while still being operationally under-specified for the learner.

\item The theorem is not a statement about all possible generalization. It is a statement about reference-free generalization. Practical success is possible when the concept class has structure. For example, if the target is constant on the off-span sector, then the collapse to one value is not harmful. If the target is determined by a local observable, then the observable itself is a reference. If the target is defined on a low-dimensional manifold of feature states, then the feature map supplies the manifold and its coordinates. In all of these cases, generalization is possible because the problem contains more than a small set of unlabeled rays.
\end{itemize}

These assumptions make the theorem easy to falsify operationally: to avoid the conclusion, one must point to a resource that breaks the residual symmetry. This is a strength rather than a weakness. It converts a vague question, ``Can a QML model generalize in an exponentially large Hilbert space?'' into a sharper question, ``What reference structure or inductive bias gives semantic meaning to the directions on which it generalizes?''

\subsection{Examples and non-examples}

It is useful to separate three cases: First, a fully reference-free rule may depend only on invariant relational data among the training rays, such as transition probabilities $\Tr(\hP_i\hP_j)$ and labels $y_i$, and then output a classifier constructed covariantly from the projectors themselves. Such a rule can orient the span of the training states, but it cannot orient directions that are completely orthogonal to that span. Second, a standard quantum kernel classifier is usually referenceful, because the feature map $x\mapsto \ket{\phi(x)}$ is fixed in advance and is implemented by a particular circuit in a particular basis~\cite{schuld2019feature,havlicek2019supervised}. The kernel matrix contains relational information among the embedded training states, but the embedding circuit also supplies a coordinate convention telling the learner which physical states correspond to which classical inputs. Our theorem does not forbid generalization in this setting; rather, it says that the feature map is part of the reference structure that makes such generalization meaningful. Third, a variational classifier with a fixed ansatz and measurement is also referenceful. The gates, connectivity, and final observable define preferred axes in Hilbert space. Such choices may be excellent inductive biases. They may also be essential for trainability and generalization~\cite{larocca2022group,meyer2023exploiting,nguyen2024equivariant}. The point of the theorem is that these choices should not be treated as empty background: they are physical or architectural structures that break the unitary gauge freedom left by the data alone.

\subsection{The stabilizer of the training data}

Given a dataset $D$, define its stabilizer group by
\begin{eqnarray}
G_D:=\{\hV \in U(\cH):\hV \hP_i \hV^\dagger = \hP_i,\; \forall i=1,\ldots,m\}.
\label{eq:stabilizer}
\end{eqnarray}
The group $G_D$ consists of all Hilbert-space rotations that leave the training data exactly unchanged. These rotations are not distinguishable by the training examples. If a reference-free learner produces the predictions that are changed under such a rotation, it would creates an orientation not present in the data.

Let
\begin{eqnarray}
S_D:=\operatorname{span}\{\ket{\psi_i}:i=1,\ldots,m\}, ~~ r_D:=\dim S_D.
\label{eq:span-rank}
\end{eqnarray}
The number $r_D$ is the reference rank of the dataset. It is bounded by the number of distinct training directions, and hence by $m$. If $r_D<d$, every unitary of the form
\begin{eqnarray}
\hV=\hPi_D\oplus \hW_{D^\perp} \quad \bigl(\hW_{D^\perp}\in U(S_D^\perp) \bigr)
\label{eq:complement-unitaries}
\end{eqnarray}
belongs to $G_D$, where $\hPi_D$ is the projector onto $S_D$. Thus, an entire unitary group remains unbroken on the unspanned sector.

\begin{figure*}[t]
\includegraphics[width=0.70\textwidth]{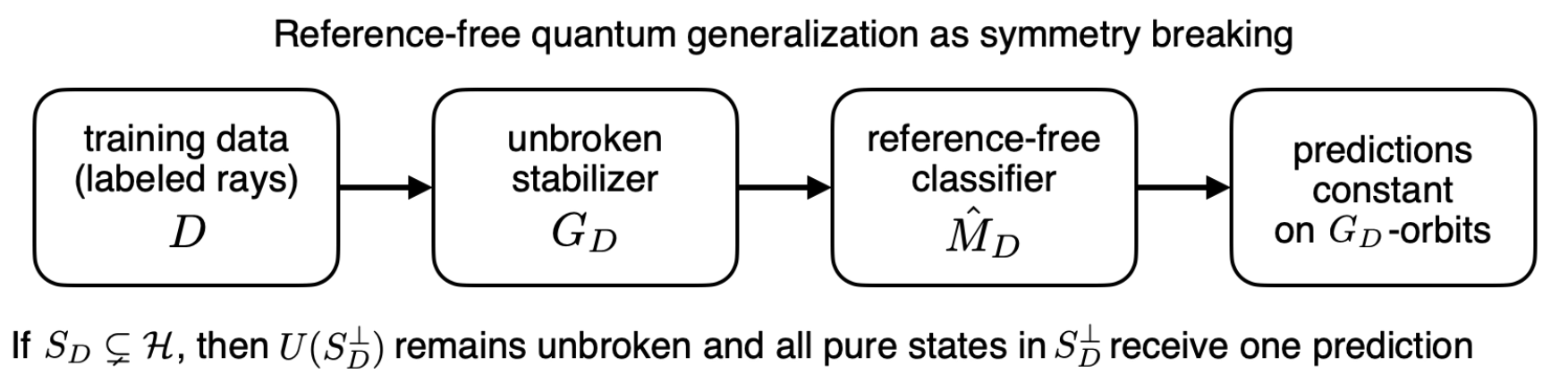}
\caption{Reference-free quantum generalization as symmetry breaking. The training data $D$ orient only the subspace and directions they physically contain. Any unitary symmetry $G_D$ left unbroken by the training data must also be left unbroken by a reference-free classifier. If the training span $S_D$ is a proper subspace of $\cH$, the full unitary group on $S_D^\perp$ remains unbroken, forcing all pure off-span test states to receive the same prediction.}
\label{fig:no-go-map}
\end{figure*}

The reference rank should be distinguished from the number of copies. If a learner receives many copies of a small number of training directions, those copies may improve the estimation accuracy, but they do not orient new directions. In the idealized projector model used here, the estimation has already been made perfect; the only remaining question is whether the physical directions have been fixed. The quantity $r_D$ measures that orientation content. Fig.~\ref{fig:no-go-map} summarizes this logic. The training data break part of the global unitary symmetry; the unbroken stabilizer then determines which test states must receive identical predictions.

\section{Main no-go theorem}\label{sec:main}

\subsection{Stabilizer-orbit invariance}

\begin{theorem}[Stabilizer-orbit no-go theorem]\label{thm:stabilizer}
Let $D\mapsto \hM_D$ be a reference-free learner in the sense of Definition~\ref{def:reference-free}. Then, for every $\hV\in G_D$,
\begin{eqnarray}
\hM_D=\hV\hM_D\hV^\dagger.
\label{eq:commuting-effect}
\end{eqnarray}
Consequently, for every test state $\hrho$,
\begin{eqnarray}
p_D(1|\hV\hrho\hV^\dagger)=p_D(1|\hrho), \quad \forall \hV \in G_D.
\label{eq:orbit-invariance}
\end{eqnarray}
\end{theorem}

\begin{proof}---If $\hV \in G_D$, then $\hV D\hV^\dagger=D$. By reference-free covariance,
\begin{eqnarray}
\hM_D=\hM_{\hV D\hV^\dagger}=\hV\hM_D\hV^\dagger,
\label{eq:proof-stabilizer-1}
\end{eqnarray}
which proves Eq.~\eqref{eq:commuting-effect}. Hence,
\begin{eqnarray}
p_D(1|\hV\hrho\hV^\dagger) &=&\Tr(\hM_D\hV\hrho\hV^\dagger) = \Tr(\hV^\dagger\hM_D\hV\hrho) \nonumber\\
	&=& \Tr(\hM_D\hrho)=p_D(1|\hrho).
\label{eq:proof-stabilizer-2}
\end{eqnarray}
The proof is completed.
\end{proof}

The theorem is elementary, but its interpretation is strong. The classifier learned from $D$ cannot assign different predictions to states in the same $G_D$ orbit. Therefore, a target concept $c:\proj(\cH)\to\{0,1\}$ is exactly learnable from $D$ by a reference-free learner only if it is constant on the stabilizer orbits that remain after training.

\begin{corollary}[Orbit obstruction to exact concept learning]\label{cor:orbit-concept}
Suppose there exist pure states $\hP$ and $\hQ$ and a unitary $\hV \in G_D$ such that $\hQ = \hV \hP \hV^\dagger$, but a target concept satisfies $c(\hP) \neq c(\hQ)$. Then, no reference-free learner can implement this target concept exactly from the dataset $D$.
\end{corollary}

This is the first form of the no-go theorem. It is not about finite sample noise, because it holds even if the training data are known exactly. It is not about limited computational power, because no restriction has been placed on the training algorithm. It is a symmetry constraint. The target function itself may be perfectly well-defined to an external observer who has a basis in mind; the statement is that this basis has not been supplied to the reference-free learner by the dataset.

\subsection{Off-span blindness}

The most transparent consequence appears when the training states do not span the full Hilbert space. Then, the stabilizer contains the full unitary group acting on the orthogonal complement. This leaves no directional invariant inside that complement. The only scalar quantity available to a two-outcome classifier is a single number.

\begin{theorem}[No off-span quantum generalization]\label{thm:offspan}
Let $D\mapsto\hM_D$ be reference-free and let $S_D$ be the span of the training vectors. If $r_D=\dim S_D<d$, then there exists a number $\lambda_D \in [0,1]$, such that
\begin{eqnarray}
\hM_D=\hPi_D\hM_D\hPi_D+\lambda_D\hPi_D^\perp,
\label{eq:offspan-block}
\end{eqnarray}
where $\hPi_D$ projects onto $S_D$ and $\hPi_D^\perp=\Id-\hPi_D$. In particular, for every pure state $\ket{\phi}\in S_D^\perp$,
\begin{eqnarray}
p_D(1|\hP_\phi)=\lambda_D.
\label{eq:offspan-constant}
\end{eqnarray}
\end{theorem}

\begin{proof}---Every unitary $\hV=\hPi_D\oplus\hW$ with $\hW\in U(S_D^\perp)$ belongs to $G_D$. By Theorem~\ref{thm:stabilizer}, $\hM_D$ commutes with all such $\hV$. Write $\hM_D$ in block form with respect to $\cH=S_D\oplus S_D^\perp$:
\begin{eqnarray}
\hM_D=\begin{pmatrix}
A & B\\
C & E
\end{pmatrix}.
\label{eq:block-decomposition}
\end{eqnarray}
The equality $(\hPi_D \oplus \hW)\hM_D = \hM_D(\hPi_D \oplus \hW)$ gives
\begin{eqnarray}
B\hW=B,
\quad
\hW C=C,
\quad
E\hW=\hW E
\label{eq:block-conditions}
\end{eqnarray}
for every $\hW\in U(S_D^\perp)$. Choosing $\hW=e^{i\theta}\Id_{S_D^\perp}$ with $e^{i\theta} \neq 1$ forces $B=C=0$. The last condition says that $E$ belongs to the commutant of the defining representation of $U(S_D^\perp)$; since this representation is irreducible, Schur's lemma gives $E=\lambda_D\Id_{S_D^\perp}$. The positivity and $\hM_D \le \Id$ imply $\lambda_D \in [0,1]$. This gives Eq.~\eqref{eq:offspan-block}, and Eq.~\eqref{eq:offspan-constant} follows immediately.
\end{proof}

The stabilizer twirl and its explicit block form are detailed in Appendix~\ref{app:twirl}.

The theorem says that the whole unspanned sector is invisible except for one scalar number. The classifier may learn a complicated effect on $S_D$, but on $S_D^\perp$ it has no directional resolution at all. In other words, the off-span sector is not merely poorly sampled; it is still gauge-equivalent under a symmetry that the training data have not broken.

This gives a clean example separating our no-go theorem from state discrimination. Let $\ket{a},\ket{b}\in S_D^\perp$ with $\braket{a}{b}=0$. A measurement in the basis containing $\ket{a}$ and $\ket{b}$ can distinguish them perfectly. Nevertheless, by Eq.~\eqref{eq:offspan-constant}, a reference-free learner trained on $D$ must satisfy
\begin{eqnarray}
p_D(1|\hP_a)=p_D(1|\hP_b)=\lambda_D.
\label{eq:orthogonal-same}
\end{eqnarray}
The difficulty is not distinguishing the states; the difficulty is that the labels of those directions have not been oriented by the training data.

A useful way to phrase the theorem is the following. Quantum mechanics permits a measurement that distinguishes $\ket{a}$ from $\ket{b}$, but the learning problem has not specified why that measurement should be the classifier. A referenceful learner may choose such a measurement because the basis is built into the hardware, the feature map, or the task description. A reference-free learner has no such right. It can only use the reference frame supplied by the data, and the data do not contain $\ket{a}$ or $\ket{b}$ as oriented directions.

\subsection{A simple worked example}\label{subsec:worked-example}

Consider the smallest example where an off-span sector has internal directions: a qutrit Hilbert space $\cH=\mathbb{C}^3$ with one training state $\hP_1=\ketbra{0}{0}$. Then, $S_D=\operatorname{span}\{\ket{0}\}$ and $S_D^\perp=\operatorname{span}\{\ket{1},\ket{2}\}$. The stabilizer contains every unitary of the form
\begin{eqnarray}
\hV=\ket{0}\!\bra{0}\oplus \hW, \quad \hW\in U(2).
\label{eq:qutrit-stabilizer}
\end{eqnarray}
A reference-free classifier must commute with all such $\hV$. Hence, it has the form
\begin{eqnarray}
\hM_D = a\ketbra{0}{0} + \lambda\left(\ketbra{1}{1} + \ketbra{2}{2}\right),~(0 \le a, ~\lambda \le 1).
\label{eq:qutrit-effect}
\end{eqnarray}
For any normalized vector in the complement,
\begin{eqnarray}
\ket{\phi(\theta,\varphi)}=\cos\theta\ket{1}+e^{i\varphi}\sin\theta\ket{2},
\label{eq:qutrit-complement-state}
\end{eqnarray}
one obtains
\begin{eqnarray}
\Tr\bigl(\hM_D \ketbra{\phi(\theta,\varphi)}{\phi(\theta,\varphi)}\bigr) = \lambda,
\label{eq:qutrit-constant}
\end{eqnarray}
independent of both angles. Thus, $\ket{1}$ and $\ket{2}$, which are orthogonal, receive the same prediction.

An external observer may say that the concept is $c(\ket{1})=0$ and $c(\ket{2})=1$. This is a perfectly valid referenceful concept, because the observer has chosen the basis $\{\ket{0},\ket{1},\ket{2}\}$. But the dataset containing only $\ket{0}$ has not supplied that basis inside the two-dimensional complement. The unitary that swaps $\ket{1}$ and $\ket{2}$ leaves the training data unchanged. A reference-free learner cannot assign different labels to the two states without using a basis that came from somewhere else.

This example captures the whole theorem in finite-dimensional form. The unspanned complement is not empty space; it may contain many mutually orthogonal and perfectly distinguishable states. What it lacks is an orientation relative to the labels.

\subsection{Approximate reference-free learning}

The realistic learning rules may not be exactly covariant. A weak external reference, noise, calibration bias, or architecture-dependent preference may break the symmetry. The exact theorem has a stable approximate version.

\begin{definition}[Stabilizer covariance defect]\label{def:defect}
For a learned effect $\hM_D$, define its covariance defect on the stabilizer of $D$ by
\begin{eqnarray}
\delta_D := \sup_{\hV \in G_D}\norm{\hM_D - \hV \hM_D \hV^\dagger}_\infty.
\label{eq:defect}
\end{eqnarray}
\end{definition}

The number $\delta_D$ measures how much the final classifier breaks symmetries that the training data themselves do not break. If $\delta_D=0$, the classifier is exactly reference-free on the stabilizer. If $\delta_D$ is small, then the classifier contains only a weak reference leakage along the unbroken directions.

\begin{theorem}[Robust orbit no-go theorem]\label{thm:robust}
If $\delta_D\le\eta$, then for every test state $\hrho$ and every $\hV \in G_D$,
\begin{eqnarray}
\abs{p_D(1 | \hV \hrho \hV^\dagger) - p_D(1 | \hrho)} \le \eta.
\label{eq:robust-orbit}
\end{eqnarray}
If, in addition, $S_D\neq\cH$, then there exist an effect $A_D$ on $S_D$ and a scalar $\lambda_D \in [0,1]$, such that
\begin{eqnarray}
\norm{\hM_D - \left( A_D \oplus \lambda_D \Id_{S_D^\perp}\right)}_\infty \le \eta.
\label{eq:robust-block}
\end{eqnarray}
Consequently, for every pure $\ket{\phi} \in S_D^\perp$,
\begin{eqnarray}
\abs{p_D(1 | \hP_\phi) - \lambda_D} \le \eta.
\label{eq:robust-offspan}
\end{eqnarray}
\end{theorem}

The proof is given in Appendix~\ref{app:robust}. It uses the Haar average over the unbroken group $U(S_D^\perp)$. The conclusion is intuitive: small symmetry breaking can create only small prediction variation along the directions that were symmetry-equivalent in the ideal reference-free limit. Thus, the exact theorem is not a fragile artifact of the perfect covariance. It has a quantitative form that can be interpreted as a trade-off between the reference leakage and the off-span prediction variation.

\section{Exponential reference-rank barrier for generic quantum concepts}\label{sec:lower-bound}

The stabilizer theorem becomes a QML lower bound when applied to generic concept classes in a $d=2^n$ dimensional Hilbert space. The relevant resource is not merely the number of examples, but the rank of the quantum reference frame they span.

\subsection{A basis-label concept class}

Fix an orthonormal basis
\begin{eqnarray}
\cB=\{\ket{e_1}, \ldots, \ket{e_d}\}
\label{eq:basis}
\end{eqnarray}
and consider binary concepts that assign labels to the basis states:
\begin{eqnarray}
y=(y_1,\ldots,y_d) \in \{0,1\}^d, \quad c_y(\hP_{e_j}) = y_j.
\label{eq:basis-labels}
\end{eqnarray}
The test distribution is uniform over the basis projectors $\{\hP_{e_j}\}_{j=1}^d$. This concept class is deliberately generic. It contains no smoothness, locality, or low-complexity structure. Its role is to expose what happens if one tries to treat the whole Hilbert space as an unstructured quantum feature space.

Suppose that the training set reveals labels on a subset $I\subseteq\{1,\ldots,d\}$ and that the corresponding training states span
\begin{eqnarray}
S_I=\operatorname{span}\{\ket{e_i}:i\in I\}, \quad r=|I|.
\label{eq:training-subset}
\end{eqnarray}
Even if the learner memorizes all training labels, Theorem~\ref{thm:offspan} forces one common prediction on every basis state in $S_I^\perp$.

\begin{theorem}[Reference-rank lower bound]\label{thm:rank-lower}
Consider the basis-label concept class in Eq.~\eqref{eq:basis-labels} with uniform test distribution. Let a reference-free learner be trained on $r$ basis directions and suppose $r<d$. Then, among balanced labelings of the unobserved set $I^c$, there exists a target labeling for which every deterministic threshold classifier obtained from $\hM_D$ has test error at least
\begin{eqnarray}
\Risk_D \ge \frac{1}{d}\left\lfloor\frac{d-r}{2}\right\rfloor.
\label{eq:rank-lower-bound}
\end{eqnarray}
Equivalently, to guarantee uniform error $\Risk_D \le \eps$ over this generic balanced concept class, one must have
\begin{eqnarray}
r \ge d-2 \eps d-1.
\label{eq:rank-requirement}
\end{eqnarray}
For $n$ qubits, $d=2^n$, so constant-error generic reference-free learning requires reference rank $r=\Omega(2^n)$.
\end{theorem}

\begin{proof}---By Theorem~\ref{thm:offspan}, every unobserved basis state $\ket{e_j}$ with $j \in I^c$ receives the same probability $\lambda_D$. A deterministic threshold rule therefore assigns the same predicted label to all $d-r$ unobserved basis states. If the true labels on $I^c$ are balanced, at least $\lfloor(d-r)/2\rfloor$ of them disagree with this constant prediction. Since the test distribution is uniform over all $d$ basis states, Eq.~\eqref{eq:rank-lower-bound} follows. Solving Eq.~\eqref{eq:rank-lower-bound} $\le\eps$ gives Eq.~\eqref{eq:rank-requirement} up to the displayed integer rounding.
\end{proof}

The lower bound is a reference-rank bound. If the $m$ training states have rank $r_D\le m$, then $m=\Omega(d)$ is necessary for this generic task. However, the statement is more precise than a sample lower bound: many copies of the same low-rank training directions cannot orient the missing subspace. Conversely, $r$ carefully chosen independent training directions can orient an $r$-dimensional sector even if the number of copies per direction is small in the idealized model.

The preceding statement used a deterministic threshold rule, because this is the most common way to turn a probability into a binary decision. A similar obstruction holds directly for probabilistic predictions. If a learner assigns the same probability $\lambda_D$ to all unobserved basis states and the unseen labels are exactly balanced, then the average absolute loss on the unseen sector is independent of $\lambda_D$:
\begin{eqnarray}
\frac{1}{d}\sum_{j \in I^c}\abs{\lambda_D - y_j} = \frac{d-r}{2d}.
\label{eq:absolute-loss-bound}
\end{eqnarray}
For squared loss, the best choice is $\lambda_D=1/2$, giving
\begin{eqnarray}
\frac{1}{d}\sum_{j \in I^c}(\lambda_D-y_j)^2 \ge \frac{d-r}{4d}.
\label{eq:square-loss-bound}
\end{eqnarray}
Thus, the obstruction is not an artifact of thresholding. The off-span sector permits only one scalar prediction, and one scalar cannot encode an arbitrary balanced binary string. The corresponding soft-loss derivation is given in Appendix~\ref{app:soft-loss}.

\subsection{Minimax interpretation}

The balanced-label condition in Theorem~\ref{thm:rank-lower} is not essential; it is a clean way to express the worst-case obstruction. The unseen sector $I^c$ is a set on which the reference-free learner must output one constant label after thresholding. For any fixed constant prediction, an adversary can choose labels on $I^c$ so that about half of them are wrong. Thus the minimax error contributed by the unspanned sector is proportional to its dimension.

If $d-r$ is large, the learner's performance on the unseen sector is no better than a constant guess for an unstructured binary string. This is exactly what one should expect from the reference-frame viewpoint. The data have not oriented the individual basis directions in $S_I^\perp$, so the labels of those directions are not physical information available to a reference-free rule. The exponential scaling for $n$ qubits follows simply because the number of unstructured directions is exponential.

The theorem also clarifies why a polynomial number of random examples can be useful for structured concept classes but not for arbitrary labels. A structured concept may be constant on large stabilizer orbits or may be determined by a small number of physically meaningful observables. In that case, the target function does not require independent labels for all $d$ basis states. The generic class in Eq.~\eqref{eq:basis-labels}, by contrast, asks for independent semantic meaning in every basis direction. Without a reference that orients those directions, no reference-free learner can infer the missing labels.

\subsection{What the lower bound does and does not say}

Theorem~\ref{thm:rank-lower} does not say that QML is impossible. It says that generic Hilbert-space labels are not learnable without a reference structure. The practical QML avoids the theorem by supplying structure, such as, feature map, computational basis, known measurement, locality, Hamiltonian, ansatz symmetry, etc~\cite{huang2021power,huang2022advantage,caro2023out,nguyen2024equivariant}. These resources reduce the effective concept class. They tell the learner which directions matter and how labels should vary between examples. In this sense, QML generalization is not powered by Hilbert-space dimension alone; it is powered by the combination of Hilbert-space dimension with physical reference structure.

This view also clarifies why structured QML generalization bounds can be much better than the generic bound above. The bounds based on trainable gates, effective dimension, Fisher information, or data symmetries quantify the capacity of a specified model class~\cite{caro2022generalization,abbas2021power,haug2024generalization}. Our theorem addresses the complementary question: before capacity can help, what reference structure makes the target function a well-defined learnable object?

There is no contradiction between the reference-rank barrier and successful quantum kernel methods. In kernel QML, the feature map defines a structured family of states indexed by classical data. The learner is not asked to label arbitrary rays in $\cH$; it is asked to generalize on the image of a particular map. That image, together with the prescribed kernel, is a reference structure. The theorem says that if one removes that structure and keeps only a few labeled rays, the rest of Hilbert space does not become a meaningful feature space by default.

\subsection{Structured escape routes}\label{subsec:escape-routes}

It is helpful to list explicitly how structured QML tasks avoid the generic lower bound. The first escape route is a low-dimensional data manifold. If all test states are promised to lie in a known manifold $\mathcal{M}\subset\proj(\cH)$, then the relevant question is not whether the learner can label all of $\proj(\cH)$, but whether the training data and the known manifold structure orient $\mathcal{M}$. A feature map from classical data is the standard example~\cite{schuld2019feature,havlicek2019supervised,huang2021power}.

The second escape route is a known observable. Suppose the target label is determined by the sign of $\Tr(\hO\hrho)$ for a fixed observable $\hO$. Then, the observable is already a reference frame. A learner that is allowed to use $\hO$ is not reference-free in the sense of Definition~\ref{def:reference-free}; it has been told which direction in operator space matters~\cite{huang2022advantage}.

The third escape route is a symmetry prior. In geometric or equivariant QML, one often knows that the task is invariant or equivariant under a group action~\cite{larocca2022group,meyer2023exploiting,nguyen2024equivariant}. This prior reduces the effective hypothesis space. It should not be confused with the stabilizer $G_D$ in our theorem. A task symmetry is a physical or semantic structure supplied by the problem. The stabilizer $G_D$ is instead the residual gauge freedom left by missing information. The former can help generalization; the latter blocks unsupported label variation.

The fourth escape route is a sufficiently high-rank training set. If the data themselves contain many independent directions, they can act as a quantum reference frame. This is the regime in which the reference rank $r_D$ approaches $d$. The exponential lower bound says that, for completely generic labels on $n$ qubits, this is expensive: one needs rank of order $2^n$. For structured labels, the required rank may be much smaller.

These escape routes are not loopholes. They are the physical content of QML. A useful quantum learning algorithm should exploit a structure that is present in the task. The no-go theorem says only that an unstructured Hilbert space does not supply that structure for free.

\subsection{Why a training Gram matrix is not enough}\label{subsec:gram-not-enough}

Many QML methods are naturally expressed in terms of Gram matrices~\cite{schuld2019feature,havlicek2019supervised,thanasilp2024exponential}. Given training states, one can form
\begin{eqnarray}
K_{ij}=\Tr(\hP_i\hP_j)=\abs{\braket{\psi_i}{\psi_j}}^2.
\label{eq:gram-training}
\end{eqnarray}
This matrix contains useful relational information among the training examples. It is invariant under global unitaries and therefore is legitimate data for a reference-free learner. However, $K_{ij}$ contains no information about directions orthogonal to $S_D$. If $\ket{\phi}\in S_D^\perp$, then
\begin{eqnarray}
\Tr(\hP_\phi\hP_i)=0 \quad \forall i,
\label{eq:offspan-zero-kernel}
\end{eqnarray}
so every pure off-span state has the same kernel vector with respect to the training set. From the kernel viewpoint, Theorem~\ref{thm:offspan} says that all such states are indistinguishable as test features unless an external feature map or test-state reference gives additional coordinates.

This observation is often hidden in classical-data kernel learning because a test input $x$ comes with a known method for computing $K(x,x_i)$. The feature map is part of the problem. In the present reference-free setting, a generic quantum test ray in $S_D^\perp$ has the same relationship to every training ray. The training Gram matrix therefore cannot support arbitrary label variation inside the complement.

The same point applies to quantum neural networks viewed as feature-space models. A high-dimensional Hilbert space may contain many possible decision boundaries, but the data-supported relational information may occupy a much smaller sector. A model can exploit the remaining dimensions only if its architecture, measurement, or feature map supplies a reference for them. Otherwise, those dimensions are gauge degrees of freedom rather than learnable features.

\section{Numerical illustrations}\label{sec:numerics}

We include simple illustrations to make the geometry visible. The simulations are finite-dimensional linear-algebra experiments designed to display the consequences of stabilizer twirling and reference-rank deficiency.

\begin{figure}[t]
\includegraphics[width=1.00\columnwidth]{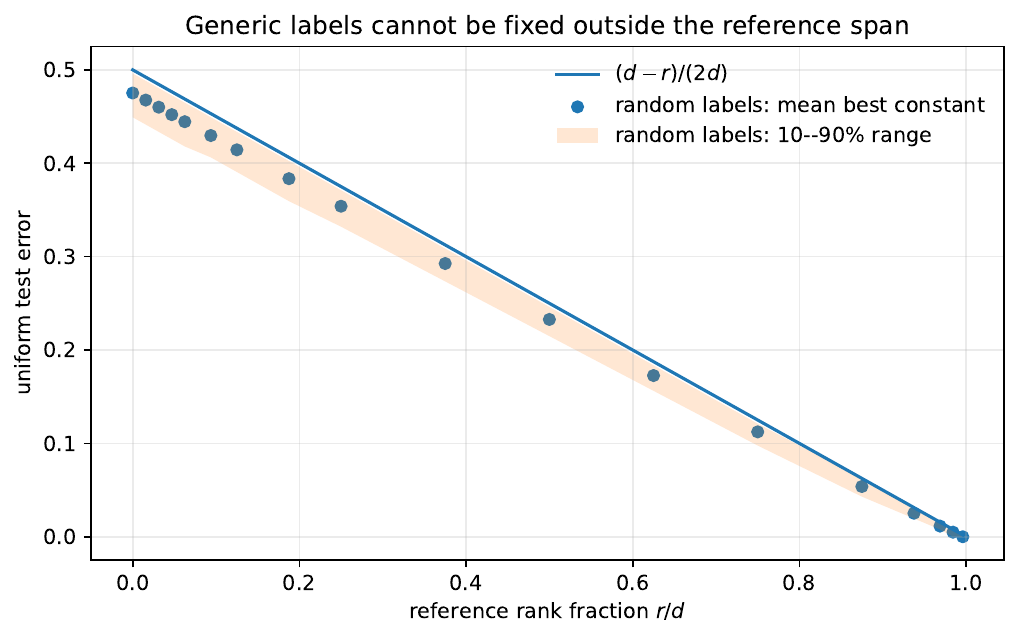}
\caption{Reference-rank barrier for generic labels. For $d=256$, a reference-free learner that has oriented only rank $r$ cannot assign independent labels in the remaining $d-r$ dimensional sector. The balanced-label lower bound scales as $(d-r)/(2d)$. Random-label points show the mean best-constant error on the unseen sector, and the shaded region gives the $10$--$90$ percentile range.}
\label{fig:rank-lower}
\end{figure}

{\em Rank-dependent error for generic labels.}---Fig.~\ref{fig:rank-lower} shows the rank law in Theorem~\ref{thm:rank-lower} for $d=256$. We evaluate several ranks $r$, and the solid curve is
\begin{eqnarray}
E_{\rm bal}(r)=\frac{d-r}{2d}.
\label{eq:app-balanced-curve}
\end{eqnarray}
This is the asymptotic balanced-label lower bound without integer rounding. For the random-label points, each trial draws labels independently from $\{0,1\}$ and assumes the learner memorizes the first $r$ labels. The remaining $d-r$ labels are classified by the better of the two constant labels. Thus the plotted random-label error in a trial is
\begin{eqnarray}
E_{\rm rand}=\frac{1}{d}\min\{N_0,N_1\},
\label{eq:app-random-error}
\end{eqnarray}
where $N_0$ and $N_1$ are the numbers of zeros and ones in the unseen sector. The number of trials is $4000$. This a posteriori choice makes the points slightly lower than the balanced worst-case line for finite $d$, but the same linear dependence on the unspanned fraction is visible.

The shaded band shows the $10$--$90$ percentile range over random labels. Its width shrinks as the unobserved sector becomes smaller. The important feature is the deterministic slope: every missing reference direction contributes one potentially independent unseen label that the reference-free learner cannot orient.

\begin{figure}[t]
\includegraphics[width=1.00\columnwidth]{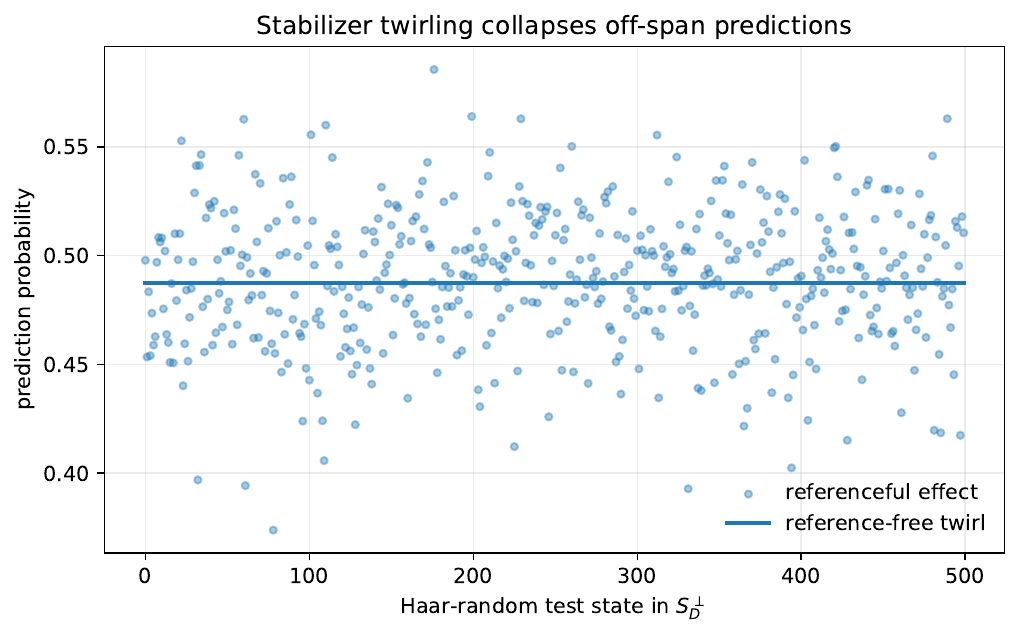}
\caption{Off-span prediction collapse. A referenceful effect can vary over pure states in $S_D^\perp$, but the stabilizer-twirled reference-free effect is scalar on $S_D^\perp$. Thus all pure off-span test states receive exactly the same prediction. The plotted example uses $d=32$, $r=8$, and $500$ Haar-random pure states in the complement.}
\label{fig:offspan-collapse}
\end{figure}

{\em Collapse of off-span predictions.}---Fig.~\ref{fig:offspan-collapse} illustrates Theorem~\ref{thm:offspan}. We set $d=32$ and $r=8$, choose a Haar-random decomposition $\cH=S_D\oplus S_D^\perp$, and sample a random effect $\hM$. The reference-free version is the stabilizer twirl
\begin{eqnarray}
\hM_{\rm twirl}=\hPi_D\hM\hPi_D+\lambda\hPi_D^\perp,
\quad
\lambda=\frac{\Tr(\hPi_D^\perp\hM)}{d-r}.
\label{eq:app-twirl}
\end{eqnarray}
The predictions are evaluated on $500$ Haar-random pure states in $S_D^\perp$. The generic referenceful effect $\hM$ gives varying predictions, whereas the twirled reference-free effect is scalar on $S_D^\perp$, so all off-span predictions collapse to one horizontal line. This is the numerical picture of Eq.~\eqref{eq:offspan-constant}.

The figure should be read as a contrast between two physical situations. The blue cloud corresponds to an effect that contains a preferred basis in the complement. Such an effect can resolve off-span directions, but it is referenceful. The horizontal line corresponds to the effect after the unbroken gauge has been removed by stabilizer twirling. It is the only structure that remains compatible with the data alone.

\begin{figure}[t]
\includegraphics[width=1.00\columnwidth]{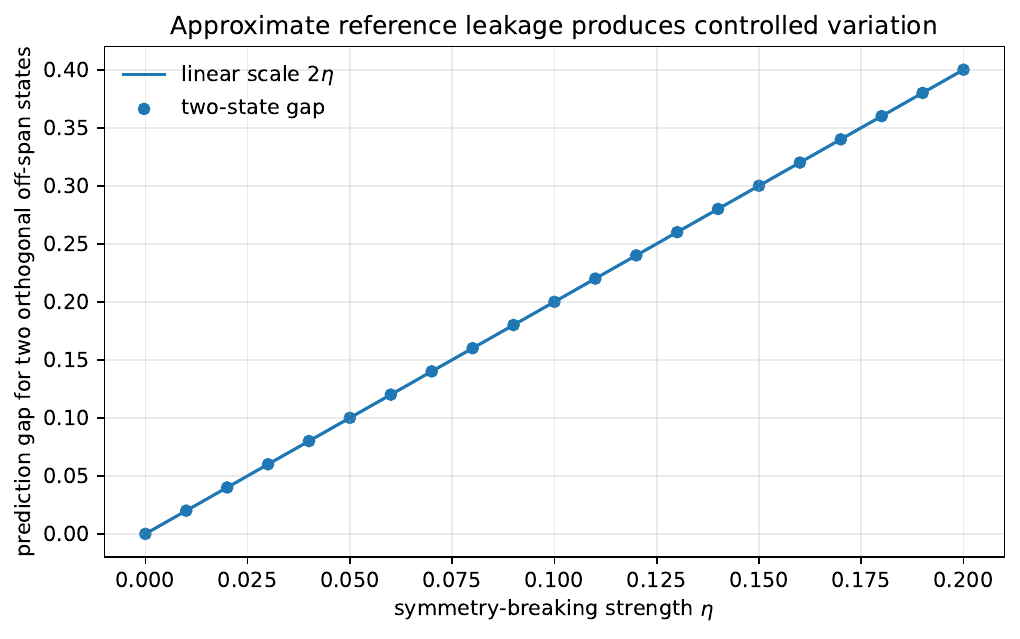}
\caption{Approximate symmetry breaking. A perturbation of strength $\eta$ inside $S_D^\perp$ creates a nonzero prediction gap between two orthogonal off-span states, but the gap remains controlled by the symmetry-breaking scale, consistent with Theorem~\ref{thm:robust}.}
\label{fig:covariance-defect}
\end{figure}

{\em Approximate reference leakage.}---Fig.~\ref{fig:covariance-defect} illustrates the robust theorem. A traceless Hermitian perturbation is inserted inside $S_D^\perp$, creating a weak external orientation in the complement. Two orthogonal complement states are chosen as eigenvectors with opposite eigenvalues, so the prediction gap is linear in the perturbation strength, consistent with the operator-norm bound in Theorem~\ref{thm:robust}. This is the approximate version of the central message: directional generalization outside the training span requires symmetry breaking, and the amount of variation is controlled by the amount of reference leakage.

\begin{figure}[t]
\includegraphics[width=1.00\columnwidth]{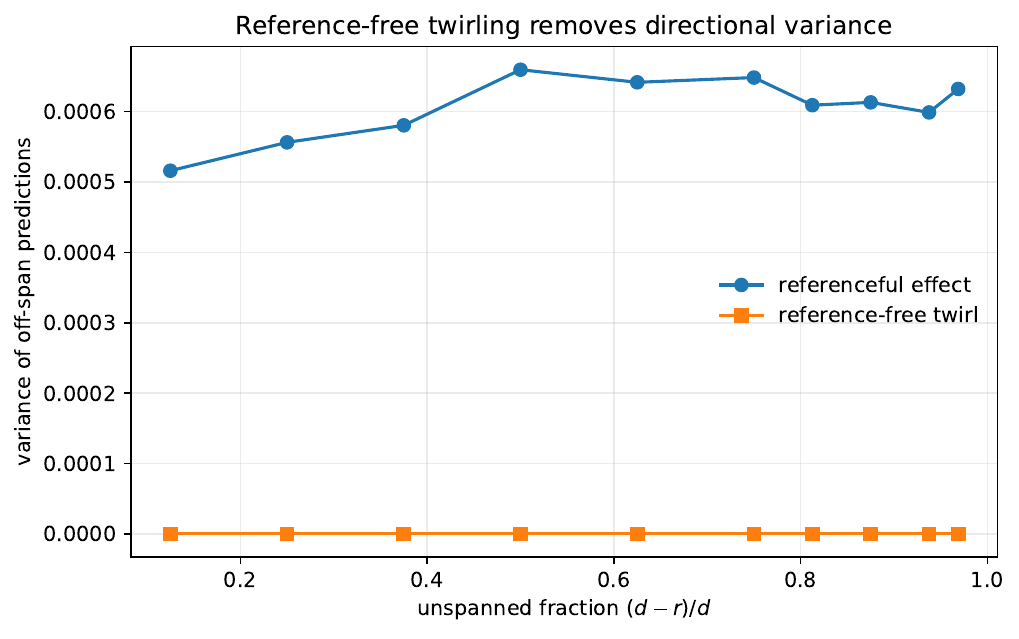}
\caption{Reference-free twirling removes directional variance. For each rank $r$ in a $d=64$ dimensional Hilbert space, a random referenceful effect has nonzero prediction variance over states in $S_D^\perp$. Stabilizer twirling removes this variance, because the twirled effect is scalar on the complement.}
\label{fig:variance-rank}
\end{figure}

{\em Directional variance removed by twirling.}---Fig.~\ref{fig:variance-rank} gives a second view of the same geometry. We set $d=64$ and sample random effects for several values of the reference rank $r$. For each $r$, we evaluate the variance of predictions on Haar-random pure states in $S_D^\perp$ before and after stabilizer twirling. The referenceful effect has nonzero directional variance, while the twirled reference-free effect has zero variance up to numerical precision because Eq.~\eqref{eq:explicit-twirl} makes the complement block scalar. The plot emphasizes that the no-go theorem is not about a particular pair of states; the entire off-span distribution collapses.

\section{Relation to existing no-go and QML results}\label{sec:relations}

\subsection{State discrimination}

Unknown-state discrimination asks whether one can identify which state was prepared from a specified set. Our theorem asks a different question: can a learner assign semantic labels to directions that the training data have not oriented? The answer can be negative even for orthogonal states. Eq.~\eqref{eq:orthogonal-same} shows that two orthogonal test states in $S_D^\perp$ must receive the same reference-free prediction. Once a measurement basis distinguishing them is supplied, the states are easy to discriminate. The no-go is therefore not a distinguishability no-go, but a label-reference no-go.

This distinction is important for interpreting the theorem. State discrimination concerns the physical distinguishability of preparations relative to a given measurement task. Reference-free generalization concerns whether the learning problem has supplied the measurement task itself. Orthogonality solves the first problem but not the second. The missing object is not a measurement outcome; it is the reference that specifies which measurement should be associated with the labels.

\subsection{Barren plateau}

Barren plateau results show that, for certain parameterized quantum circuits, the gradients can become exponentially small with system size, making variational training difficult~\cite{mcclean2018barren,cerezo2021variational}. Our theorem has no landscape. It applies after training, to any final POVM effect that a reference-free learning rule could output. A learner may have perfect optimization and still be unable to assign the labels along a stabilizer orbit.

Conversely, an architecture can avoid barren plateaus and still be subject to the reference-frame issue if it is reference-free and the data leave a large stabilizer. The two obstructions live at different levels. The barren plateaus concern the ability to find parameters. The present theorem concerns what the output classifier is allowed to depend on once parameters have been found.

\subsection{Average no-free-lunch theorem}

Quantum no-free-lunch theorems study average performance limits over ensembles of tasks or processes~\cite{poland2020nofree,sharma2022reformulation,wang2024transition}. The present theorem is dataset-wise. Once a particular $D$ is fixed, the stabilizer $G_D$ determines exactly which directions remain gauge-equivalent. This gives a deterministic obstruction before any averaging over task distributions.

The reference-rank lower bound in Sec.~\ref{sec:lower-bound} has a no-free-lunch flavor because it considers a generic unstructured concept class. Nevertheless, its mechanism is different. It is not an averaged theorem over all functions. It is the consequence of a stabilizer symmetry left unbroken by a fixed dataset. The lower bound then follows because an adversarial or balanced labeling can vary inside an orbit on which the learner is forced to be constant.

\subsection{Connection to quantum reference frames}

The closest physical analogy is the theory of quantum reference frames and asymmetry resources~\cite{bartlett2007reference,gour2008resource}. In that setting, a lack of reference frame imposes covariance or superselection-like restrictions. Here, the same idea is applied to QML generalization. The training data act as a finite quantum reference frame for label structure. If their span is incomplete, a unitary reference freedom remains in the complement. External features, measurements, Hamiltonians, or architecture choices are precisely the resources that break this freedom.

The novelty is the identification of supervised quantum generalization as a reference-frame problem. A training set does two jobs: it provides examples and it breaks symmetry. Generalization outside the data-supported frame is possible only to the extent that some other physical structure has already broken the remaining symmetry.

\subsection{Connection to QML capacity and geometry}

Modern QML theory studies model capacity using quantities such as trainable gates, effective dimension, and data-dependent quantum Fisher information~\cite{caro2022generalization,abbas2021power,haug2024generalization}. These works ask how a specified model generalizes once the input representation and measurement structure are fixed~\cite{banchi2021generalization,caro2023out,gilfuster2024rethinking}. Our theorem asks a prior question: what makes the labels of quantum states physically meaningful in the first place? The two viewpoints are complementary. A model may have high capacity, but if it is reference-free and the data have a large stabilizer, that capacity cannot be used to break the unbroken symmetry.

This distinction is particularly relevant when discussing overparameterized quantum models. A highly expressive ansatz may be able to represent many effects on $\cH$, including effects that separate arbitrary orthogonal states in $S_D^\perp$. But expressibility is not the same as identifiability from reference-free data. If two effects differ only by how they act inside an unbroken orbit, the data alone cannot select between them. Thus the theorem supplies an identifiability constraint complementary to capacity estimates: a model may represent a classifier without the reference-free training problem specifying which representative of that classifier family should be chosen.

Symmetry-aware QML provides a useful comparison~\cite{larocca2022group,meyer2023exploiting,nguyen2024equivariant}. There, one intentionally builds invariances or equivariances into the model because they are known symmetries of the task. In the present theorem, the stabilizer symmetry is different: it is not necessarily a desirable task symmetry, but a residual gauge symmetry caused by missing reference information. A good architecture may deliberately break it by supplying a feature map or basis, or may exploit it if the target concept is genuinely constant on those orbits.

\subsection{Practical QML interpretation}

The theorem gives a diagnostic question for QML proposals: what physical structure supplies the label reference? In a classical-data quantum model, the answer may be the feature map. In a quantum-data classifier for phases of matter, the answer may be locality, a Hamiltonian family, or a symmetry class~\cite{cong2019qcnn}. In a hardware-efficient variational classifier, the answer may be the native gate set and final measurement. In a kernel method, the answer may be the embedding and the kernel relation among inputs~\cite{thanasilp2024exponential}.

This diagnostic explains why they are necessary. A QML model should not merely point to the dimension of Hilbert space; it should specify the reference structure that makes the relevant directions meaningful and learnable. The no-go theorem identifies the cost of omitting this structure.

\section{Discussion and outlook}\label{sec:discussion}

We have proved a no-go theorem for reference-free quantum generalization. If a learner has no external Hilbert-space orientation, then it must be covariant under global unitary changes of reference frame. This simple physical requirement has strong consequences. The learned classifier must commute with the stabilizer of the training data; hence predictions are constant on stabilizer orbits. If the training states do not span the full Hilbert space, then the entire orthogonal complement is collapsed to a single prediction value. Generic labels in that complement are therefore unlearnable without additional reference information.

The central message of this work is that quantum learning requires a physical reference frame. In practical QML, such frames are everywhere: computational basis states, Pauli measurements, data-encoding maps, locality, Hamiltonians, symmetries, or ansatz architectures. These resources are often treated as background choices. Our theorem shows that they are not merely implementation details; they are part of what makes generalization possible.

The central conceptual shift is to view training data as the symmetry-breaking resources. A dataset has two kinds of content. It has statistical content, namely labels and frequencies, and it has orientational content, namely the Hilbert-space directions that it physically fixes. The sample-complexity analyses usually focus on the first. The present theorem isolates the second. If the orientational content is low-rank, then a large part of Hilbert space remains gauge-equivalent and cannot support independent labels.

There is also a practical suggestiveness. When reporting a QML advantage or a QML generalization claim, one should identify the reference structure in the same way one identifies the ansatz and the training distribution. For classical data, this may seem automatic because the data coordinates are given. For quantum data, the coordinates may themselves be physical resources. A statement such as ``the learner generalizes to unseen quantum states'' is incomplete until the family of unseen states and the reference connecting them to the training states have been specified.

This work therefore supports a conservative view of Hilbert-space dimension. The dimension $2^n$ is a reservoir of possible representations, not an automatic source of learnable labels. The useful part of that reservoir is selected by the physical structures of the problem: encodings, observables, symmetries, dynamics, and training examples. In this sense, the result of this work is compatible with both optimism and caution about QML. It allows strong advantages when the reference structure is rich and quantumly accessible, but it rules out a naive advantage based only on the formal size of the state space.

The broad lesson is that the exponential size of Hilbert space is not a free feature space. It is a potential feature space whose directions become meaningful only through physical structure. Quantum generalization is, in this precise sense, a form of symmetry breaking.

\begin{acknowledgments}
This work was supported by the IITP grant funded by the Korean government (RS-2019-II190003) and the Korean ARPA-H Project through KHIDI, funded by the Ministry of Health \& Welfare, Republic of Korea (RS-2025-25456722). We thank the Yonsei University Quantum Computing Project Group for support and access to Quantum System One.
\end{acknowledgments}

\onecolumngrid
\appendix

\section{Proof of the robust theorem}\label{app:robust}

Here, we prove Theorem~\ref{thm:robust}. Suppose
\begin{eqnarray}
\delta_D=\sup_{\hV \in G_D} \norm{\hM_D - \hV \hM_D \hV^\dagger}_\infty \le \eta.
\label{eq:app-defect}
\end{eqnarray}
For any state $\hrho$ and $\hV\in G_D$,
\begin{eqnarray}
\abs{p_D(1|\hV\hrho\hV^\dagger)-p_D(1|\hrho)} = \abs{\Tr\left[(\hV^\dagger\hM_D\hV - \hM_D) \hrho\right]} \le \norm{\hV^\dagger\hM_D\hV-\hM_D}_\infty\Tr\hrho \le \eta,
\label{eq:app-robust-orbit}
\end{eqnarray}
which proves Eq.~\eqref{eq:robust-orbit}. The only ingredients are cyclicity of trace, positivity of $\hrho$, and the definition of the operator norm.

Now assume $S_D \neq \cH$. Let $T:=S_D^\perp$ and average over the subgroup $\{\hPi_D \oplus \hW : \hW \in U(T)\} \subseteq G_D$:
\begin{eqnarray}
\overline{M}_D := \int_{U(T)}(\hPi_D \oplus \hW)\hM_D(\hPi_D \oplus \hW)^\dagger d\Haar(\hW).
\label{eq:haar-average}
\end{eqnarray}
Because $\overline{M}_D$ is a convex average of effects, it is also an effect. Moreover,
\begin{eqnarray}
\norm{\hM_D - \overline{M}_D}_\infty &=& \norm{\int_{U(T)}\left[\hM_D - (\hPi_D\oplus\hW)\hM_D(\hPi_D\oplus\hW)^\dagger\right]d\Haar(\hW)}_\infty \nonumber\\
	&\le& \int_{U(T)}\norm{\hM_D-(\hPi_D\oplus\hW)\hM_D(\hPi_D\oplus\hW)^\dagger}_\infty d\Haar(\hW)
\le\eta.
\label{eq:average-close}
\end{eqnarray}
The same block-commutant argument as in Theorem~\ref{thm:offspan} gives
\begin{eqnarray}
\overline{M}_D = A_D \oplus \lambda_D \Id_T
\label{eq:average-block}
\end{eqnarray}
for some effect $A_D$ on $S_D$ and some $\lambda_D\in[0,1]$. Eqs.~\eqref{eq:average-close} and \eqref{eq:average-block} prove Eq.~\eqref{eq:robust-block}. For a pure state $\ket{\phi}\in T$,
\begin{eqnarray}
\abs{p_D(1|\hP_\phi)-\lambda_D}
=\abs{\bra{\phi}(\hM_D-\overline{M}_D)\ket{\phi}}
\le\norm{\hM_D-\overline{M}_D}_\infty\le\eta,
\label{eq:app-robust-offspan}
\end{eqnarray}
which proves Eq.~\eqref{eq:robust-offspan}.

\section{The stabilizer twirl and the block form}\label{app:twirl}

The stabilizer twirl used above is the conditional expectation onto the algebra of operators commuting with the unbroken group. For $T=S_D^\perp$, define
\begin{eqnarray}
\mathcal{T}_D(X) = \int_{U(T)}(\hPi_D \oplus \hW) X (\hPi_D \oplus \hW)^\dagger d\Haar(\hW).
\label{eq:twirl-map}
\end{eqnarray}
Writing $X$ in block form,
\begin{eqnarray}
X
=
\begin{pmatrix}
A & B \\ 
C & E 
\end{pmatrix},
\label{eq:app-block-X}
\end{eqnarray}
one obtains
\begin{eqnarray}
\mathcal{T}_D(X)
=
\begin{pmatrix}
A & 0 \\
0 & \frac{\Tr E}{\dim T} \Id_T
\end{pmatrix}.
\label{eq:explicit-twirl}
\end{eqnarray}
The off-diagonal blocks vanish because $\int_{U(T)}\hW d\Haar(\hW)=0$. The lower-right block becomes scalar because the Haar average is invariant under conjugation by every unitary on $T$. Eq.~\eqref{eq:explicit-twirl} is the explicit algebraic form behind the off-span blindness theorem.

This expression also explains the numerical figures. A referenceful effect may have a nontrivial lower-right block $E$, producing different predictions for different vectors in $T$. The reference-free twirl replaces $E$ by its normalized trace. Thus, every pure state in $T$ receives the same value, independent of the original eigenvectors of $E$.

\section{Randomized learners and finite-copy remarks}\label{app:randomized}

A learning algorithm may use internal randomness. Let $\omega$ denote its random seed and let $\hM_{D,\omega}$ be the effect output on that seed. The operational prediction is the averaged probability
\begin{eqnarray}
\bar{p}_D(1|\hrho)=\mathbb{E}_\omega\Tr(\hM_{D,\omega}\hrho)=\Tr(\bar M_D\hrho),
\quad
\bar M_D:=\mathbb{E}_\omega\hM_{D,\omega}.
\label{eq:randomized-average}
\end{eqnarray}
If the distribution over output effects is covariant, then the averaged effect $\bar{M}_D$ satisfies Eq.~\eqref{eq:covariance}, and all the theorems apply to the operational predictor. If covariance holds for every seed separately, then the theorems apply pathwise to each $\hM_{D,\omega}$.

The finite-copy setting introduces a different limitation. If the learner receives $N_i$ copies of each training state rather than exact projectors, then it must estimate the training rays and their span. This adds ordinary statistical error. It does not remove the stabilizer obstruction. At best, many copies allow the learner to determine $S_D$ more accurately; they do not create new independent directions outside the span of the training states. Therefore, the reference-rank barrier should be understood as an idealized lower bound that remains after ordinary estimation difficulty has been removed.

\section{Soft-loss version of the rank barrier}\label{app:soft-loss}

For completeness, we spell out Eqs.~\eqref{eq:absolute-loss-bound} and \eqref{eq:square-loss-bound}. Let $N=d-r$ be the number of unobserved basis states. Suppose that the labels on $I^c$ are balanced, so that $N_0=N_1=N/2$ when $N$ is even. Since Theorem~\ref{thm:offspan} forces the same probability $\lambda$ on every unobserved state, the absolute loss contributed by the unseen sector is
\begin{eqnarray}
\frac{1}{d}\left(N_0\abs{\lambda-0}+N_1\abs{\lambda-1}\right)
=\frac{1}{d}\frac{N}{2}\left(\lambda+1-\lambda\right)
=\frac{N}{2d}.
\label{eq:app-absolute}
\end{eqnarray}
For squared loss,
\begin{eqnarray}
\frac{1}{d}\left(N_0\lambda^2+N_1(1-\lambda)^2\right)
=\frac{N}{2d}\left[\lambda^2+(1-\lambda)^2\right]
\ge \frac{N}{4d},
\label{eq:app-square}
\end{eqnarray}
with equality at $\lambda=1/2$. If $N$ is odd, the same formulas hold up to the obvious $O(1/d)$ rounding correction. These soft-loss bounds express the same reference-rank obstruction without committing to a hard classification threshold.

\bibliographystyle{apsrev4-2}
\bibliography{ref_QML_no-go}

\end{document}